\documentclass[11pt, fleqn]{article}

\usepackage[letterpaper, margin=1in]{geometry}
\usepackage{graphicx}
\usepackage{authblk}

\usepackage{hyperref} 
\usepackage{cite}
\usepackage{amsmath,amssymb,amsthm,mathtools,mathrsfs,bm,booktabs}
\usepackage{fullpage}
\usepackage{enumitem}
\usepackage{algorithm}
\usepackage{algorithmic}
\usepackage{xcolor}
\usepackage{cleveref,autonum}
\usepackage{tikz}
\tikzset{
  dot/.style = {circle, fill, minimum size=#1,inner sep=0pt, outer sep=0pt},
  dot/.default = 3pt,
}

\usepackage{version}
\excludeversion{conference}\includeversion{fullpaper}\usepackage[appendix=inline]{apxproof}\renewenvironment{appendixproof}[1][]{\proof}{\endproof}
\renewcommand{\appendixsectionformat}[2]{}

\newcommand{\bA}{\bm{A}}

\newcommand{\argmax}{\mathop{\rm arg\,max}}
\newcommand{\balpha}{\bm{\alpha}}
\newcommand{\bx}{\bm{x}}
\newcommand{\bp}{\bm{p}}
\newcommand{\bq}{\bm{q}}
\newcommand{\dcup}{\mathbin{\dot{\cup}}}
\renewcommand{\epsilon}{\varepsilon}

\newtheorem{theorem}{Theorem} 
\newtheorem{lemma}{Lemma}
\newtheorem{proposition}{Proposition}
\newtheorem{corollary}{Corollary}
\newtheorem{definition}{Definition}
\newtheorem{example}{Example}

\title{Fair and Efficient Balanced Allocation for Indivisible Goods}
\author[1,2]{Yasushi Kawase}
\author[1,2]{Ryoga Mahara}
\affil[1]{The University of Tokyo}
\affil[2]{RIKEN Center for Advanced Intelligence Project}
\affil[ ]{\texttt{kawase@mist.i.u-tokyo.ac.jp, mahara@mist.i.u-tokyo.ac.jp}}
\date{}

\begin{document}

\maketitle

\begin{abstract}
We study the problem of allocating indivisible goods among agents with additive valuation functions to achieve both fairness and efficiency under the constraint that each agent receives exactly the same number of goods (the \emph{balanced constraint}).
While this constraint is common in real-world scenarios such as team drafts or asset division, it significantly complicates the search for allocations that are both fair and efficient.
Envy-freeness up to one good (EF1) is a well-established fairness notion for indivisible goods.
Pareto optimality (PO) and its stronger variant, fractional Pareto optimality (fPO), are widely accepted efficiency criteria.
Our main contribution establishes both the existence and polynomial-time computability of allocations that are simultaneously EF1 and fPO under balanced constraints in two fundamental cases: (1) when each agent has a personalized bivalued valuation, and (2) when agents have at most two distinct valuation types,.
Our algorithms leverage novel applications of maximum-weight matching in bipartite graphs and duality theory, providing the first polynomial-time solutions for these cases and offering new insights for constrained fair division problems.
\end{abstract}

\section{Introduction}
The fair division of \emph{indivisible} goods is an important problem that has been widely studied in mathematics, economics, and computer science~\cite{brams1996fair,brandt2016handbook}.
In recent years, this topic has attracted even more interest (see surveys~\cite{amanatidis2023fair,aziz2022algorithmic,guo2023survey}).
Many previous studies have discussed fairness and efficiency without allocation constraints.
However, in real-world problems, it is often necessary to make allocations under various constraints.
Motivated by this, recent works have studied fair division problems under a variety of constraints (see the survey~\cite{suksompong2021constraints}).

For instance, in team sports drafts, new players are assigned to teams, and it is typical for each team to receive the same number of new players to prevent any team from having a numerical advantage.
Another example is the division of indivisible assets, such as inherited jewelry or artwork, among family members. Siblings often agree to take the same number of items, but differences in market value or personal significance may still lead to envy. These examples highlight the importance of considering allocations that maintain balanced quantities while also satisfying fairness and efficiency.

In this paper, we study the problem of finding fair and efficient allocations under the constraint that each agent receives the same number of goods (the \emph{balanced constraint}), assuming that the total number of goods is a multiple of the number of agents.
We also assume that each agent has an additive valuation function.

To formalize notions of fairness in such settings, several criteria have been introduced in the literature.
\emph{Envy-freeness} (EF)~\cite{foley1966resource} is one of the most fundamental fairness concepts.
It requires that no agent prefers someone else's bundle over their own.
However, with indivisible items, achieving EF can be impossible even in very simple scenarios, such as when there are only two agents and a single item.
To address this issue, various relaxed notions of fairness have been proposed.
The most prominent among these is \emph{envy-freeness up to one item} (EF1) introduced by Budish~\cite{budish2011combinatorial}.
EF1 requires that each agent prefers their own bundle to that of any other agent after removing at most one item from the other agent's bundle.
It is known that an EF1 allocation always exists and can be computed in polynomial time~\cite{lipton2004approximately}.
Moreover, even under the balanced constraint, an EF1 allocation can be found using the round-robin algorithm~\cite{caragiannis2019unreasonable}.

On the efficiency side, \emph{Pareto optimality} (PO) is a standard benchmark.
An allocation is PO if no agent can be made strictly better off without making someone else worse off.
A stronger concept, known as \emph{fractional Pareto optimality} (fPO), extends this notion to fractional allocations.
It is known that maximizing utilitarian welfare yields an fPO allocation, and thus also a PO allocation. Consequently, such an allocation always exists and can be computed in polynomial time.

While finding allocations that are either EF1 or PO individually is relatively straightforward, the problem becomes considerably more complex when both fairness and efficiency must be achieved at the same time. 
In what follows, we examine the main challenges in satisfying both criteria simultaneously and review previous progress on addressing this issue.

For the unconstrained case, Caragiannis et al.~\cite{caragiannis2019unreasonable} established the novel result that maximizing the Nash social welfare~\cite{Nash50, kaneko1979nash}, defined as the geometric mean of the agents' valuations, yields an allocation that is both EF1 and PO (but may not be fPO). However, since maximizing Nash social welfare is NP-hard~\cite{nguyen2014computational} and even APX-hard~\cite{lee2017apx}, this approach does not directly yield an efficient algorithm.
Barman et al.~\cite{barman2018finding} addressed this computational limitation by proposing a pseudo-polynomial time algorithm that computes an EF1 and PO allocation and proved the existence of an EF1 and fPO allocation.
Subsequently, Mahara~\cite{mahara2024polynomial} developed a polynomial-time algorithm to compute an EF1 and fPO allocation when the number of agents is constant.
The existence of a polynomial-time algorithm for finding an EF1 and PO (or fPO) allocation in general case remains an important open question.

There is growing interest in fair and efficient allocation under various practical constraints that extend beyond the unconstrained setting.
When such constraints are present, achieving both EF1 and PO\footnotemark{} becomes more challenging than in the unconstrained case. 
A natural first approach is to maximize Nash social welfare among allocations that satisfy the constraints, but in general, such an approach does not lead to an EF1 allocation.
Furthermore, under general constraints, there may not exist any allocation that is both EF1 and PO~\cite{wu2021budget,cookson2025constrained}.
To the best of our knowledge, the only existing result guaranteeing the existence of an allocation that is both EF1 and PO under constraints is due to \cite{shoshan2023efficient}, which studies category constraints and gives a polynomial-time algorithm for the case of two agents.

\footnotetext{In the constrained case, we consider PO and fPO only with respect to feasible allocations (i.e., constrained PO and constrained fPO). It is not difficult to see that PO (or fPO) allocations with respect to all allocations may not exist. For example, suppose there are two agents and two goods, and one agent values both goods positively while the other agent assigns zero value to both. Then, any PO allocation with respect to all allocations must assign both goods to the agent who values them, but such an allocation is not balanced.}

\subsection*{Our Results}
In this paper, we study the problem of fair division of indivisible goods among agents with additive valuation functions under the \emph{balanced constraints}. 

We first observe that maximizing Nash social welfare over balanced allocations does not necessarily guarantee EF1 (see \Cref{ex:simple}).
We then give a characterization that an allocation is fPO if and only if it maximizes a weighted utilitarian social welfare under balanced constraints (\Cref{prop:fPO}).
We also show that checking whether a given balanced allocation satisfies fPO and EF1 is polynomial-time solvable, whereas checking PO is coNP-complete. 
Additionally, we provide a reduction from the problem of finding an EF1 and fPO allocation under unconstrained case to balanced constraint case.

We establish both the existence and polynomial-time computability of balanced allocations that are simultaneously EF1 and fPO in two important cases: when each agent has a personalized bivalued valuation, i.e., each agent $i$ assigns two distinct values $a_i,b_i~(a_i>b_i\ge 0)$ to the goods, and when there are at most two types of agents. 

In the personalized bivalued case, we show that maximizing the weighted utilitarian social welfare with appropriately chosen weights yields an EF1 and fPO allocation.

In the two types case, we provide an algorithm that maintains primal and dual optimal solutions for the linear program that maximizes weighted social welfare. The algorithm searches for an EF1 allocation by gradually changing the weights. 
In this procedure, we exploit a characterization of dual optimal solutions for maximum-weight perfect matchings in bipartite graphs using the shortest path distances. It is worth mentioning that our result for the two types case implies the result for the unconstrained case.

It is worth noting that the balanced constraints are a special case of the category constraints studied in~\cite{shoshan2023efficient,igarashi2025fair} and the matroid constraints considered in~\cite{KS2020,kawase2023random,cookson2025constrained}.
Moreover, both cases analyzed in this paper have been widely studied in the recent fair division literature (see Related Work).

\begin{conference}
Due to space limitations, we omit most of the proofs, which can be found in the full paper.
\end{conference}

\subsection*{Related Work}
\paragraph{Fair and efficient allocation for divisible items}
The Fisher market framework, originally introduced by Irving Fisher (see~\cite{brainard2005compute}), has long been a central object of study in both economics and computer science.
This model with divisible goods is well known for exhibiting strong notions of fairness and efficiency.
In particular, Varian\cite{varian1974equity} showed that when all agents have equal budgets, the equilibrium allocation achieved in a Fisher market is both EF and PO.
It is established that such market equilibria can be computed in polynomial time under additive valuation functions~\cite{devanur2008market, orlin2010improved, vegh2012strongly}.
For the constrained case, an EF and PO allocation still exists if each agent has an identical constraint~\cite{cole2021existence}, but finding it is PPAD-complete even in the balanced setting where each agent receives exactly one good~\cite{trobst2024cardinal,caragiannis2024complexity}.
Moreover, there exists an allocation that is both SD-EF (which is a stronger notion than EF and balancedness) and ordinally efficient (which is slightly weaker notion than PO)~\cite{kojima2009}; however, such an allocation may not exist even under category constraint~\cite{kawase2023random}.

\paragraph{Fair and efficient allocation for indivisible chores}
For indivisible chores, the existence of EF1 and fPO (or PO) allocations was previously known for a few restricted cases:
instances with two agents~\cite{aziz2022fair};
instances with bivalued valuations for each agent~\cite{ebadian2022fairly, garg2022fair};
instances with two types of chores~\cite{Aziz2023};
instances with three agents~\cite{garg2023new};
and instances with three types of valuation functions~\cite{garg2024weighted}.
In a significant breakthrough, Mahara~\cite{mahara2025existence} recently showed the existence of EF1 and fPO allocations for general additive valuations.

\paragraph{Fair and efficient allocation under constraints}

Shoshan et al.~\cite{shoshan2023efficient} proposed a polynomial-time algorithm for fair division under category constraints with two agents, which finds an EF1 and PO allocation when each category consists of only goods or only chores. If goods and chores are mixed, the algorithm finds an EF[1,1] (envy-free up to one good and one chore) and PO allocation.
Igarashi and  Meunier~\cite{igarashi2025fair} extended this result to general settings with $n$ agents, proving the existence of a PO allocation in which each agent can be made envy-free by reallocating at most $n(n-1)$ items.

For budget constraints, Wu et al.~\cite{wu2021budget} showed that any budget-feasible allocation that maximizes the Nash social welfare achieves a $1/4$-EF1 and PO allocation for goods.
They also showed that there exists an instance in which there is no $(1/2+\epsilon)$-EF1 and PO allocation for any $\epsilon > 0$.
Here, $\alpha$-EF1 is an approximate relaxation of EF1, which requires that each agent prefers their own bundle to $\alpha$ times the bundle of any other agent after the removal of at most one item from that other agent's bundle.
Cookson et al.~\cite{cookson2025constrained} investigated fair division under matroid constraints, showing that maximizing the Nash social welfare yields a $1/2$-EF1 and PO allocation for goods.

For a broader overview of fair division under various constraints, see the survey by Suksompong~\cite{suksompong2021constraints}.



\paragraph{Bivalued and two-type instances}
For bivalued instances\footnote{In a bivalued instance, each agent $i$ assigns one of two fixed values, $a > b > 0$, to the goods. 
This setting is a special case of the personalized bivalued instance, where each agent may assign their own pair of distinct values.},
Amanatidis et al.~\cite{amanatidis2021} showed that an EFX allocation can be computed in polynomial time for goods.
Here, EFX (envy-freeness up to any good) is a stronger fairness concept than EF1.
Garg and Murhekar~\cite{garg2023computing} further showed that both EFX and PO allocations can be computed in polynomial time for goods. In the context of chores, Garg et al.~\cite{garg2022fair} and Ebadian et al.~\cite{ebadian2022fairly} proved that EF1 and PO allocations can be computed in polynomial time. Additionally, Garg et al.~\cite{garg2023new} showed that EFX and PO allocations can be computed in polynomial time for chores when there are three agents. 
Akrami et al.~\cite{akrami2022maximizing} established that maximizing the Nash social welfare is polynomial-time computable when the two values are multiples of each other, but becomes NP-hard when they are coprime.

For instances with two types of valuations, Mahara~\cite{mahara2023existence,mahara2024extension} showed that an EFX allocation always exists for goods under additive valuations, and even under more general monotone valuations. In addition, Garg et al.~\cite{garg2023new} showed that EF1 and PO allocations can be computed in polynomial time for both goods and chores.

\section{Model}
For each natural number $\ell$, we denote $[\ell]=\{1,\ldots,\ell\}$.

An instance of our problem is a tuple $(N,M,(v_i)_{i\in N})$, where
$N=[n]$ represents the (non-empty) set of agents and $M=[m]$ represents the (non-empty) set of \emph{indivisible} goods.
Each agent $i$ has a \emph{valuation function}, denoted as $v_i\colon M \to \mathbb{R}_+$, where $\mathbb{R}_+$ represents the set of nonnegative real numbers.
We assume that each agent's valuation is \emph{additive}, and write $v_i(X)\coloneqq \sum_{j\in X}v_{ij}$ to denote the utility of agent $i$ when $i$ receives a subset of goods $X\subseteq M$.
A valuation function $v_i$ is called \emph{personalized bivalued} if there exist $a_i,b_i\in\mathbb{R}_+$ with $a_i>b_i$ such that $v_{ij}\in\{a_i,b_i\}$ for all $j\in M$.

Throughout this paper, we assume that $m$ is a multiple of $n$, and define $k\coloneqq m/n$.
An (integral) \emph{allocation} is an ordered partition $\bA=(A_1,\dots,A_n)$ of $M$, i.e., $\bigcup_{i\in N}A_i=M$ and $A_i\cap A_{i'}=\emptyset$ for any distinct $i,i'\in N$.
Each $A_i$ is the \emph{bundle} allocated to agent $i$.
An allocation $\bA$ is called \emph{balanced} if $|A_i|=k$ for all $i\in N$.
We sometimes represent a balanced allocation $\bA$ by a matrix $x\in \{0,1\}^{N\times M}$, where $x_{ij}=1$ if good $j$ is allocated to agent $i$ (i.e., $j \in A_i$), and $x_{ij}=0$ otherwise.
A \emph{balanced fractional allocation} is a matrix $x\in \mathbb{R}_+^{N\times M}$ such that $\sum_{j\in M}x_{ij}=k$ for all $i\in N$ and $\sum_{i\in N}x_{ij}=1$ for all $j\in M$.
By the property of the total unimodularity, the set of balanced fractional allocations forms the convex full of the balanced (integral) allocations~(see, e.g., Schrijver~\cite[Sec. 21.2]{schrijver2003}).
Thus, a balanced fractional allocation can be interpreted as a lottery over balanced allocations.

An allocation $\bA$ is called \emph{envy-free up to one good (EF1)} if, for all $i,i'\in N$, either $A_{i'} = \emptyset$, or there exists a good $j\in A_{i'}$ such that $v_i(A_i)\ge v_i(A_{i'}\setminus\{j\})$.
A balanced (fractional) allocation $x\in \mathbb{R}_+^{N\times M}$ \emph{Pareto-dominates} a balanced (fractional) allocation $x'\in\mathbb{R}_+^{N\times M}$ if 
(i) $\sum_{j\in M}v_{ij}x'_{ij}\le \sum_{j\in M}v_{ij}x_{ij}$ for all $i\in N$, and 
(ii) $\sum_{j\in M}v_{ij}x'_{ij}< \sum_{j\in M}v_{ij}x_{ij}$ for some $i\in N$.
A balanced allocation is called \emph{Pareto optimal (PO)} or \emph{fractionally Pareto optimal (fPO)} if it cannot be Pareto-dominated by any other balanced integral or fractional allocation, respectively.

Observe that the set of achievable valuation vectors $(\sum_{j\in M}v_{1j}x_{1j},\dots,\sum_{j\in M}v_{nj}x_{nj})$ arising from balanced fractional allocation $x$ forms a polytope, which is precisely the convex full of the set of valuation vectors $(v_1(A_1),\dots,v_n(A_n))$ corresponding to balanced allocations $\bA$.
Thus, the fPO allocations can be characterized within the framework of multi-objective linear programming (see, e.g., \cite{ehrgott2005}).
\begin{proposition}\label{prop:fPO}
A balanced allocation $\bA$ is fPO if and only if it maximizes the weighted sum (weighted utilitarian social welfare) $\sum_{i\in N}\alpha_iv_i(A_i)$ for some positive weight vector $\balpha\in\mathbb{R}_{++}^N$.
\end{proposition}
Note that, for any given strictly positive weight vector $\balpha$, the maximum value $\sum_{i\in N}\alpha_i\cdot v_i(A_i)$ taken over all balanced allocations $\bA$ can be computed by the following linear programming (LP):
\begin{align}\label{eq:primal}
\begin{array}{rll}
\max            & \sum_{i\in N}\sum_{j\in M}\alpha_i v_{ij}x_{ij}&\\
\mathrm{s.t.}   & \sum_{i\in N}x_{ij} = 1  & \forall j\in M, \\
                & \sum_{j\in M}x_{ij} = k  & \forall i\in N, \\
                & x_{ij}            \ge 0  & \forall i\in N,\ \forall j \in M.
\end{array}
\end{align}

To illustrate the above concepts, we present a concrete example.

\begin{example}\label{ex:simple}
Consider an instance with $N=\{1,2\}$, $M=\{1,2,3,4\}$, and the valuations are
\begin{align}
&(v_{11},v_{12},v_{13},v_{14})=(10,10,21,22),\ \text{and}\\
&(v_{21},v_{22},v_{23},v_{24})=(0,1,6,8).
\end{align}
The possible valuation vectors for balanced allocations are visualized in \Cref{fig:example}.
This instance has a unique EF1 and fPO allocation: $(\{1,3\},\{2,4\})$.
Note that the allocation $(\{1,4\},\{2,3\})$ is PO but not fPO.
Among the balanced allocations, $(\{1,2\},\{3,4\})$ maximizes the Nash social welfare, but it is not EF1.
Moreover, in this instance, no EF fractional allocation (i.e., ex ante EF) can be represented as a lottery over EF1 and fPO allocations (i.e., ex post EF1 and fPO).
\end{example}

\begin{figure}[htbp]
\centering
\begin{tikzpicture}[scale=0.2]
    \def\xmin{18}
    \draw[very thin,black!10] (\xmin,0) grid (45,15);
    \draw[->,thick] (\xmin,0) -- (45,0) node[below right] {$v_1(A_1)$};
    \draw[->,thick] (\xmin,0) -- (\xmin,15) node[left] {$v_2(A_2)$};
    \draw[blue,thick,fill=blue,fill opacity=.1] (20,14) -- (31,9) -- (43,1) -- (32,6) -- cycle;
    \node[dot,rectangle,label={[yshift=3pt]right:\textcolor{gray}{\tiny$(\{3,4\},\{1,2\})$}}] at (43, 1) {};
    \node[dot,red,label={left:\textcolor{gray}{\tiny$(\{2,4\},\{1,3\})$}}] at (32, 6) {};
    \node[dot,red,label={left:\textcolor{gray}{\tiny$(\{2,3\},\{1,4\})$}}] at (31, 8) {};
    \node[dot,red,label={right:\textcolor{gray}{\tiny$(\{1,4\},\{2,3\})$}}] at (32, 7) {};
    \node[dot,rectangle,red,label={right:\textcolor{gray}{\tiny$(\{1,3\},\{2,4\})$}}] at (31, 9) {};
    \node[dot,rectangle,label={[yshift=3pt]right:\textcolor{black}{\footnotesize$\bA=(\{1,2\},\{3,4\})$}}] at (20,14) {};
    \foreach \x in {20,25,30,35,40} {
        \draw (\x,0) -- (\x,-0.2) node[below,black!70] {\scriptsize\x};
    }
    \foreach \y in {5,10} {
        \draw (\xmin,\y) -- (\xmin-0.2,\y) node[left,black!70] {\scriptsize\y};
    }
\end{tikzpicture}
\caption{Possible valuation vectors for balanced allocations in \Cref{ex:simple}. 
Square markers indicate fPO allocations and red markers denote EF1 allocations.
The blue region represents the set of possible valuation vectors for fractional balanced allocations.}\label{fig:example}
\end{figure}
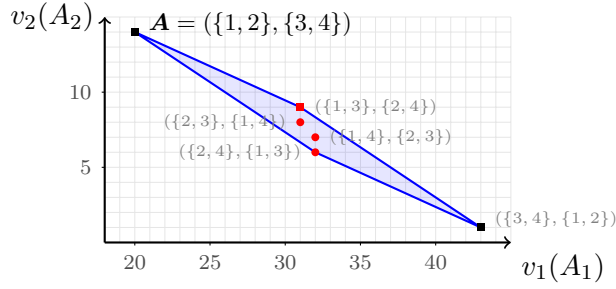

\section{Properties of EF1 and fPO Allocations}
In this section, we provide some properties of EF1 and fPO that will be used in subsequent discussions.

It is well known that an EF1 balanced allocation can be computed by the following \emph{round-robin procedure}~\cite{caragiannis2019unreasonable}:
First, an ordering of the agents is fixed. Then, according to this ordering, the agents take turns sequentially to choose their favorite available good (breaking ties arbitrarily). This process is repeated in multiple rounds until all goods have been allocated.
We call the outcome of the round-robin procedure the \emph{round-robin allocation}.



For a balanced allocation $\bA=(A_1,\dots,A_n)$ and a weight vector $\balpha\in\mathbb{R}_{++}^N$, define a weighted directed graph $G^{(\balpha)}_{\bA}=(N\dcup M\dcup\{r\};\overrightarrow{E}\cup\overleftarrow{E}\cup E^+,w^{(\balpha)})$, where $\overrightarrow{E}=N\times M$, $\overleftarrow{E}=\{(j,i)\in M\times N\mid j\in A_i\}$, $E^+=\{r\}\times M$, $w^{(\balpha)}(i,j)=-\alpha_iv_{ij}$ for all $(i,j)\in \overrightarrow{E}$, $w^{(\balpha)}(j,i)=\alpha_iv_{ij}$ for all $(j,i)\in \overleftarrow{E}$, and $w^{(\balpha)}(r,j)=0$ for all $j\in M$.
It is not difficult to see that $\bA$ maximizes $\sum_{i\in N}\alpha_iv_i(A_i)$ over all balanced allocations if and only if $G^{(\balpha)}_{\bA}$ contains no negative-weight directed cycles.

Next, we characterize the optimal solution of LP~\eqref{eq:primal}.
Its dual is given by:
\begin{align}\label{eq:dual}
\begin{array}{rll}
\min          & k\cdot \sum_{i\in N}q_i + \sum_{j\in M} p_j\\
\mathrm{s.t.} & q_i+p_j\ge \alpha_i v_{ij} \quad \forall i\in N,\ \forall j\in M.
\end{array}
\end{align}
We call the dual variables $(\bq,\bp)$ the \emph{potentials} and $\bp$ the \emph{price vector}.
By the complementary slackness theorem, a feasible pair $\bx$ and $(\bq,\bp)$ forms optimal solutions to the primal and dual LPs if and only if, for all $i\in N$ and $j\in M$, either $x_{ij}=0$ or $q_i+p_j=\alpha_iv_{ij}$.

Let $\bA$ be a balanced fPO allocation that maximizes the value $\sum_{i\in N}\alpha_iv_i(A_i)$ for some $\balpha\in\mathbb{R}_{++}^N$.
Then, we can construct an optimal dual solution as follows (see, e.g., Schrijver~\cite[Sec.~17.4]{schrijver2003}).
\begin{fullpaper}
For completeness, we provide a proof in the appendix.
\end{fullpaper}
\begin{lemma}[folklore]\label{lem:optimal-dual}
Let $\bA$ be a balanced allocation that maximizes $\sum_{i\in N}\alpha_iv_i(A_i)$ for some $\balpha\in\mathbb{R}_{++}^N$. Then, an optimal dual solution $(\bq,\bp)$ to \eqref{eq:dual} can be obtained by setting $q_i$ as the shortest length of $r$--$i$ path for each $i\in N$ and $-p_j$ as the shortest length of $r$--$j$ path in $G^{(\balpha)}_{\bA}$.
Here, the solution is nonnegative, i.e., $(\bq,\bp)\in\mathbb{R}_+^N\times \mathbb{R}_+^M$.
Moreover, the potentials computed in this way do not depend on the choice of $\bA$.
\end{lemma}

We introduce price envy-freeness up to one good ($\bp$-EF1), which is a key concept introduced by Barman et al.~\cite{barman2018finding}. This notion has been widely used in the literature for computing EF1 and PO allocations~\cite{barman2018finding, ebadian2022fairly, mahara2024polynomial, garg2022fair, garg2023new, garg2024weighted,mahara2025existence}.
\begin{definition}[$\bp$-EF1]
Let $(\bq,\bp)$ be potentials.
For any non-empty set of goods $X\subseteq M$, define $p(X)\coloneqq \sum_{j\in X}p_j$ and $\hat{p}(X)\coloneqq p(X)-\max_{j\in X}p_j$.
An allocation $\bA$ is called $\bp$-EF1 if $p(A_i)\ge \hat{p}(A_{i'})$ for any pair of agents $i,i'\in N$.
\end{definition}

For an appropriate choice of prices~$\bp$, $\bp$-EF1 implies EF1.
\begin{lemma}\label{lem:pEF1}
Given a weight vector $\balpha \in \mathbb{R}_{++}^N$, 
let $\bA$ be an optimal balanced allocation of LP~\eqref{eq:primal}, and let $(\bq, \bp)$ be an optimal solution of the dual LP~\eqref{eq:dual} such that $q_i \ge 0$ for all $i \in N$. 
If $\bA$ is $\bp$-EF1, then it is also EF1.
\end{lemma}
\begin{appendixproof}[Proof of \Cref{lem:pEF1}]
For each $i\in N$ and $j\in M$, set $x_{ij}=1$ if $j\in A_i$ and $x_{ij}=0$ otherwise.
Then, $(x_{ij})_{i\in N,j\in M}$ is an optimal solution to the primal LP~\eqref{eq:primal}.
Thus, for any $i\in N$ and $j\in M$ with $j\in A_i$, we have $q_i+p_j=\alpha_iv_{ij}$.

Consider any pair of agents $i,i'\in N$.
Let $j^*$ be a good in $\argmax_{j\in A_{i'}}p_j$.
Then, we have
\begin{align}
v_i(A_i)
&=\sum_{j\in A_i}v_{ij}=\sum_{j\in A_i}\frac{q_i+p_j}{\alpha_i}
=\frac{kq_i}{\alpha_i}+\frac{p(A_i)}{\alpha_i}\\
&\ge \frac{kq_i}{\alpha_i}+\frac{\hat{p}(A_{i'})}{\alpha_i}
=\frac{kq_i}{\alpha_i}+\frac{p(A_{i'}\setminus\{j^*\})}{\alpha_i}\\
&=\frac{q_i}{\alpha_i} + \sum_{j\in A_{i'}\setminus\{j^*\}}\frac{q_i+p_j}{\alpha_i}\\
&\ge \sum_{j\in A_{i'}\setminus\{j^*\}}v_{ij}
=v_i(A_{i'}\setminus\{j^*\}).
\end{align}
This shows that if $\bA$ is $\bp$-EF1, then it is also EF1.
\end{appendixproof}

By the LP formulation, we can check whether a given balanced allocation $\bA$ is fPO in polynomial time.
\begin{theorem}\label{thm:checkEF1fPO}
Given a balanced allocation $\bA$, whether it satisfies EF1 and fPO can be checked in polynomial time.
\end{theorem}
\begin{appendixproof}[Proof of \Cref{thm:checkEF1fPO}]
Checking whether a given allocation $\bA$ is EF1 can be done in polynomial time by verifying the EF1 condition for all pairs of agents $i,i'\in N$. Thus, it remains to check whether the allocation is fPO.

We show that a balanced allocation $\bA$ is fPO if and only if the optimum value of the following LP is 0:
\begin{align}
\begin{array}{rll}
\max            & \sum_{i\in N}z_i&\\
\mathrm{s.t.}   & \sum_{j\in M}v_{ij}x_{ij}=v_i(A_i)+z_i & \forall i\in N,\\
                & \sum_{i\in N}x_{ij} = 1  & \forall j\in M, \\
                & \sum_{j\in M}x_{ij} = k  & \forall i\in N, \\
                & x_{ij},z_i        \ge 0  & \forall i\in N,\ \forall j \in M.
\end{array}
\end{align}
If $\bA$ is Pareto-dominated by a fractional balanced allocation $x$, then by setting $z_i=\sum_{j\in M}v_{ij}x_{ij}-v_i(A_i)$ for each $i\in N$, we can construct a feasible solution with a positive objective value, and thus conclude that the optimum value of the LP is positive. 
Conversely, if the optimum value of the LP is positive, then the fractional allocation appeared in the optimum solution Pareto dominates $\bA$, hence $\bA$ is not fPO.
Thus, $\bA$ is fPO if and only if the optimum value of the LP is 0.
Since linear programs can be solved in polynomial time, checking fPO can also be done in polynomial time.
\end{appendixproof}

Next, we demonstrate that any instance of the unconstrained problem can be transformed into an equivalent instance with balanced constraint.
This is achieved by introducing a set of ``dummy'' goods that have no value to any agent, allowing for a direct correspondence between the properties of allocations in the two settings. This equivalence is formalized in the following theorem.
\begin{theorem}\label{thm:unconstrained_to_balanced}
For an unconstrained fair allocation instance $(N,M,(v_i)_{i\in N})$, let $M'$ be a set of dummy goods of size $|N|\cdot (|M|-1)$ and let $v_i'(X)=v_i(X\cap M)$ for each $i\in N$ and $X\subseteq M\cup M'$.
The following equivalences hold:
\begin{itemize}
\item If an allocation $\bA$ is PO (resp., fPO, EF1) in the unconstrained instance, then a balanced allocation $(A_1\cup D_1,\dots,A_n\cup D_n)$ where $(D_1,\dots,D_n)$ is an ordered partition of $M'$ is PO (resp., fPO, EF1) in the balanced instance.
\item If a balanced allocation $\bA'$ is PO (resp., fPO, EF1) in the balanced instance, then an allocation $(A_1\cap M,\dots,A_n\cap M)$ is PO (resp., fPO, EF1) in the unconstrained instance.
\end{itemize}
\end{theorem}
This theorem implies that if we construct a polynomial-time algorithm to find an EF1 and fPO balanced allocation for a specific setting, then it can also be used to find an EF1 and fPO allocation for the corresponding unconstrained setting.
This result has immediate algorithmic implications. For instance, our polynomial-time algorithm for the two-types case under the balanced constraint can be directly applied to the two-types case in the unconstrained setting. This is because, after the addition of dummy goods, the agent types remain unchanged, ensuring that the resulting balanced instance continues to fall within the two-types framework.

It is known that, in the unconstrained setting, the problem of deciding whether a given allocation is PO is coNP-complete~\cite{de2009complexity}.
By combining this with \Cref{thm:unconstrained_to_balanced}, we obtain the following corollary.
\begin{corollary}\label{cor:checkPO}
The problem of deciding whether a given balanced allocation is PO is coNP-complete.
\end{corollary}

\section{Personalized Bivalued Valuations Case}
In this section, we consider the case of personalized bivalued valuations, where each agent $i \in N$ assigns to every good $j \in M$ a value $v_{ij} \in \{a_i, b_i\}$ for some $a_i > b_i\ge 0$. We show that a balanced allocation that is both EF1 and fPO always exists, and that such an allocation can be computed in polynomial time.

\begin{theorem}\label{thm:bivalued}
When each agent has a personalized bivalued valuation, there always exists a balanced allocation that is both EF1 and fPO. 
Moreover, such an allocation can be found in polynomial time.
\end{theorem}

Our algorithm prepares $k$ slots for each agent, and considers a perfect matching between slots and goods as a balanced allocation. We set $\epsilon=1/(nk(k+1))$ and the weight of each edge between $s$th slot of agent $i$ and good $j\in M$ as 
\begin{align}
w\big((i,s),j\big)=\begin{cases}
a_i/(a_i-b_i)+s\cdot\epsilon & \text{if }v_{ij}=a_i,\\
b_i/(a_i-b_i)                & \text{if }v_{ij}=b_i.
\end{cases}
\end{align}
Our algorithm outputs a balanced allocation $\bA^*$ corresponding to a maximum weight perfect matching $X^*$ in this weighted bipartite graph.
Intuitively, the weights without the $s\cdot\epsilon$ term are chosen so that the objective function increases by the same amount regardless of which agent receives a high-value good. The $s\cdot\epsilon$ term is included to ensure that high-value goods are distributed as evenly as possible among the agents.

\begin{fullpaper}
A summary of our algorithm is presented in \Cref{alg:bivalued}.
\begin{algorithm}[ht]
\begin{algorithmic}[1]
    \caption{Personalized bivalued case}\label{alg:bivalued}
    \STATE Let $\epsilon\gets \frac{1}{nk(k+1)}$
    \STATE Construct a complete bipartite graph $(V,M; E)$ where $V=N\times [k]$ and $E=V\times M$
    \STATE Define the weight $w(e)$ for each edge $e=\big((i,s),j\big)$ as $a_i/(a_i-b_i)+s\cdot\epsilon$ if $v_{ij}=a_i$ and $b_i/(a_i-b_i)$ if $v_{ij}=b_i$
    \STATE Let $X^*\subseteq E$ be the maximum-weight perfect matching in the bipartite graph, and let $\bA^*$ be the corresponding balanced allocation where $A^*_i=\big\{j\in M\mid \big((i,s),j\big)\in X^*\text{ for some }s\in[k]\big\}$\;
    \STATE \textbf{Return} $\bA^*$\;
\end{algorithmic}
\end{algorithm}
\end{fullpaper}

In the personalized bivalued setting, fPO allocations can be characterized by a specific weight vector $\balpha$ such that $\alpha_i=1/(a_i-b_i)$ for each $i\in N$.
This motivates the choice of this particular weight vector.
\begin{proposition}\label{prop:bivalued-fPO}
For agents with personalized bivalued valuations, a balanced allocation $\bA$ is fPO if and only if it maximizes the value $\sum_{i\in N}v_{i}(A'_i)/(a_i-b_i)$ over all balanced allocations.
\end{proposition}
\begin{appendixproof}[Proof of \Cref{prop:bivalued-fPO}]
We prove both directions of the equivalence.

The ``if'' direction is immediate from \Cref{prop:fPO}. 
For the ``only if'' direction, we proceed by contraposition. Suppose there exists a balanced allocation $\bA$ that does not maximize $\sum_{i\in N}v_{i}(A_i)/(a_i-b_i)$. We show that $\bA$ is not fPO.

Under this assumption, there exists another balanced allocation $\bA'$ such that $\sum_{i\in N}v_{i}(A'_i)/(a_i-b_i)>\sum_{i\in N}v_{i}(A_i)/(a_i-b_i)$.
Consider the symmetric difference between the corresponding matchings of $\bA$ and $\bA'$ in the complete bipartite graph $(N,M;N\times M)$, which can be decomposed into alternating cycles.
Then, there must exist an improving cycle $(i_1,j_1,i_2,j_2,\dots,i_h,j_h)$ with
\begin{enumerate}[label=(\alph*)]
    \item distinct agents $i_1,i_2,\dots,i_h\in N$,
    \item distinct goods $j_1,j_2,\dots,j_h\in M$,
    \item $j_t\in A_{i_t}$ for all $t\in[h]$, and
    \item $\sum_{t=1}^h \frac{v_{i_t j_t}}{a_{i_t}-b_{i_t}}<\sum_{t=1}^h \frac{v_{i_{t} j_{t-1}}}{a_{i_t}-b_{i_t}}$ (where $j_{0}\coloneqq j_h$).
\end{enumerate}
Define index sets 
\begin{align}
I=\{t\in[h] \mid v_{i_tj_t}<v_{i_{t}j_{t-1}}\}\text{ and }
I'=\{t\in[h] \mid v_{i_tj_t}>v_{i_{t}j_{t-1}}\}.
\end{align}
Note that $(v_{i_tj_t},v_{i_{t}j_{t-1}})=(b_{i_t},a_{i_t})$ if $t\in I$, $(v_{i_tj_t},v_{i_{t}j_{t-1}})=(a_{i_t},b_{i_t})$ if $t\in I'$, and $(v_{i_tj_t},v_{i_{t}j_{t-1}})=(a_{i_t},a_{i_t})$ or $(b_{i_t},b_{i_t})$ otherwise.
Hence, $v_{i_tj_{t-1}}-v_{i_tj_t}=a_{i_t}-b_{i_t}$ if $t\in I$, $v_{i_tj_{t-1}}-v_{i_tj_t}=-(a_{i_t}-b_{i_t})$ if $t\in I'$, and $v_{i_tj_{t-1}}-v_{i_tj_t}=0$ otherwise.
Thus, condition (d) implies
\begin{align}
|I|-|I'|
&=\sum_{t=1}^h \frac{v_{i_t j_{t-1}}-v_{i_t j_t}}{a_{i_t}-b_{i_t}}\\
&=\sum_{t=1}^h \frac{v_{i_t j_{t-1}}}{a_{i_t}-b_{i_t}}-\sum_{t=1}^h \frac{v_{i_t j_{t}}}{a_{i_t}-b_{i_t}}>0.
\end{align}

We now consider two cases:
\begin{description}
\item[Case 1: $|I'|=0$.]
Construct allocation $\hat{\bA}$ by transferring each good $j_t$ from $i_t$ to $i_{t+1}$ (modulo $h$, so $j_h$ is transferred to $i_1$).
In this case, no agent is worse off after the transfer since every agent receives a good with at least as much value as the one they gave away, and each agent in $I~(\neq \emptyset)$ strictly improves their utility. Thus, $\hat{\bA}$ Pareto dominates $\bA$, implying that $\bA$ is not fPO.

\item[Case 2: $|I'|>0$.]
Let $I=\{i_{t_1},i_{t_2},\dots,i_{t_\ell}\}$ where $t_1<t_2<\dots<t_\ell$.
Since $|I|>|I'|>0$, we have $\ell\ge 2$.
There exists a cyclic interval $[t_s,t_{s+1}]$ containing no $t\in I'$. 
Here, $[t_s,t_{s+1}]=\{t_s,t_s+1,\dots,t_{s+1}\}$ for $s\in[\ell-1]$ and 
$[t_\ell,t_{\ell+1}]=[t_\ell,t_1]=\{t_\ell,t_\ell+1,\dots,h,1,2,\dots,t_1\}$.
Construct allocation $\tilde{\bA}$ by rotating goods within the cycle $(i_{t_s},j_{t_s},i_{t_s+1},j_{t_s+1},\dots,i_{t_{s+1}},j_{t_{s+1}})$.
Then, $\tilde{\bA}$ Pareto dominates $\bA$ since no agent is worse off after the rotation, and agent $i_{t_s}$ is strictly better off.
This is because, 
\begin{itemize}
    \item agents $i_{t_s+1},\dots,i_{t_{s+1}-1}$ are not in $I\cup I'$, and hence they remain unchanged;
    \item agent $i_{t_s}\in I$ does not worse off since $i_{t_s}$ gives up good $j_{t_{s}}$ of value $b_{i_{t_s}}$ and gets good $j_{t_{s+1}}$ of value at least $b_{i_{t_s}}$;
    \item agent $i_{t_{s+1}}\in I$ strictly better off since $i_{t_{s+1}}$ gives up good $j_{t_{s+1}}$ of value $b_{i_{t_{s+1}}}$ and gets good $j_{t_{s+1}-1}$ of value $a_{i_{t_{s+1}}}$.
\end{itemize}
Thus, $\tilde{\bA}$ Pareto dominates $\bA$, implying that $\bA$ is not fPO.
\end{description}

In both cases, we achieve a Pareto improvement over $\bA$. Therefore, $\bA$ cannot be fPO.
\end{appendixproof}

We show that $\bA^*$ computed in our algorithm satisfies both EF1 and fPO.
\begin{lemma}\label{lem:bivalued-EF1fPO}
$\bA^*$ is EF1 and fPO.
\end{lemma}
\begin{appendixproof}[Proof of \Cref{lem:bivalued-EF1fPO}]
First, we show that $\bA^*$ is EF1.
Suppose to the contrary that $\bA^*$ is not EF1.
Then, there exist agents $i,i'\in N$ such that 
\begin{align}
v_i(A^*_i)<\min_{j\in A^*_{i'}}v_{i}(A^*_{i'}\setminus\{j\}). \label{eq:notEF1}
\end{align}
Let $s=|\{j\in A^*_i\mid v_{ij}=a_i\}|$ and $s'=|\{j\in A^*_{i'}\mid v_{ij}=a_i\}|$. 
By \eqref{eq:notEF1} and the fact that $\bA^*$ is balanced, we have $s'\ge s+2$. 
Let $j_0 \in A_i^*$ be a good with $v_{ij_0}=b_i$ that maximizes the slot index $s_0$ such that $((i,s_0),j_0)\in X^*$.
Also, let $j'_0 \in A_{i'}^*$ be a good with $v_{ij'_0}=a_i$ that minimizes the slot index $s'_0$ such that $((i,s'_0),j'_0)\in X^*$.
Note that such goods exist. Indeed, $k\ge s'\ge s+2$ implies $k-s>0$ and $s'>0$.

By the choice of $s_0$ and $s'_0$, we have $s_0\ge k-s$ and $s'_0\le k-s'+1$.
Thus, we get
\begin{align}
s_0\ge k-s\ge k-(s'-2)\ge s'_0+1.
\end{align}
Consider the modified matching 
\begin{align}
\MoveEqLeft[8]
X'\coloneqq (X^*\setminus\{\big((i,s_0),j_0\big),\big((i',s'_0),j'_0\big)\})\cup\{\big((i,s_0),j'_0\big),\big((i',s'_0),j_0\big)\}.
\end{align}
Then, we have
\begin{align}
\sum_{e\in X'}w(e)-\sum_{e\in X^*}w(e) &= w((i,s_0), j'_0)-w((i,s_0), j_0) + w((i',s'_0), j_0) - w((i',s'_0), j'_0)\\
&\ge \left(\tfrac{a_i}{a_i-b_i}+s_0\cdot\epsilon\right)-\tfrac{b_i}{a_i-b_i}+\tfrac{b_{i'}}{a_{i'}-b_{i'}}-\left(\tfrac{a_{i'}}{a_{i'}-b_{i'}}+s'_0\cdot\epsilon\right)\\
&\ge s_0\cdot\epsilon - s'_0\cdot\epsilon\\ 
&\ge \epsilon \\
&> 0.
\end{align}
This contradicts the optimality of $X^*$.

%

Next, we demonstrate that $\bA^*$ is fPO.
For each balanced allocation $\bA$ and a corresponding perfect matching $X\subseteq E$, we have
\begin{align}
\sum_{i\in N}\frac{v_i(A^*_i)}{a_i-b_i}
&\ge \sum_{e\in X^*}w(e)-n\cdot \frac{k(k+1)}{2} \epsilon\\
&= \sum_{e\in X^*}w(e)-\frac{1}{2}\\
&\ge \sum_{e\in X}w(e)-\frac{1}{2}\\
&\ge \sum_{i\in N}\frac{v_i(A_i)}{a_i-b_i}-\frac{1}{2}.
\end{align}
Note that 
\begin{align}
\sum_{i\in N}\frac{v_i(A^*_i)}{a_i-b_i}-\sum_{i\in N}\frac{v_i(A_i)}{a_i-b_i}
&=\sum_{i\in N}\frac{v_i(A^*_i)-v_i(A_i)}{a_i-b_i}
\end{align}
is an integer because $\frac{v_{ij}-v_{ij'}}{a_i-b_i}\in\{0,1,-1\}$ for all $j,j'\in M$.
Hence, $\sum_{i\in N}\frac{v_i(A_i^*)}{a_i-b_i}\ge \sum_{i\in N}\frac{v_i(A_i)}{a_i-b_i}$.
Thus, by \Cref{prop:fPO}, the balanced allocation $\bA^*$ must be fPO.
\end{appendixproof}

Now we are ready to prove \Cref{thm:bivalued}.
\begin{proof}[Proof of \Cref{thm:bivalued}]
We can compute the maximum-weight perfect matching $X^*$ in the bipartite graph in polynomial time by the Hungarian algorithm, which runs in polynomial time~\cite{schrijver2003}.
From this matching, we construct the balanced allocation $\bA^*$ defined by 
$A_i^* = \{\, j \in M \mid ((i,s), j) \in X^* \,\}$, 
which can also be obtained in polynomial time. 
By \Cref{lem:bivalued-EF1fPO}, this allocation satisfies both EF1 and fPO. 
Therefore, for agents with personalized bivalued valuations, a balanced allocation that is simultaneously EF1 and fPO always exists and can be computed in polynomial time.
\end{proof}

\section{Two Types Case}
In this section, we consider instances in which agents are partitioned into at most two types, referred to as type-1 and type-2 agents. This structure serves as a tractable intermediate case between the setting of identical agents and fully heterogeneous populations.

We begin by revisiting the simplest case, where all agents belong to a single type; that is, every agent shares the same valuation function ($v_i = u$ for all $i \in N$). In this setting, the round-robin allocation is not only guaranteed to be EF1, but it also trivially satisfies fPO. 
This is because, with identical valuations, every complete allocation maximizes utilitarian social welfare.

We now turn to the case with two types. 
For each type $t \in \{1,2\}$, let $N_t$ denote the set of agents of type $t$, with $n_t = |N_t|$ and $n = n_1 + n_2$ the total number of agents. Without loss of generality, we assume $N_1 = \{1, \dots, n_1\}$ and $N_2 = \{n_1 + 1, \dots, n\}$. All agents of the same type share an identical valuation function: for each $i \in N_t$, $v_i = u_t$.

Our main result in this section is the following.
\begin{theorem}\label{thm:2types}
When the number of agent types is at most two, there always exists a balanced allocation that is both EF1 and fPO. 
Moreover, such an allocation can be found in polynomial time.
\end{theorem}

To obtain such an fPO balanced allocation, we only consider weight vectors $\balpha\in\mathbb{R}_{++}^N$ such that $\alpha_i=1$ for all $i\in N_1$ and $\alpha_i=\gamma$ for all $i\in N_2$, where $\gamma\in\mathbb{R}_{++}$.

Without loss of generality, we may assume $u_{tj} \ne u_{tj'}$ for some $t\in\{1,2\}$ and $j, j' \in M$. Otherwise, the problem is trivial since every balanced allocation is both EF and fPO.
We denote
\begin{align}
\delta \coloneqq \frac{\min_{t\in\{1,2\},\, j,j'\in M:\,u_{tj}\ne u_{tj'}}|u_{tj}-u_{tj'}|}{1+\max_{t\in\{1,2\}}\max_{j \in M} u_{tj}}~(>0). \label{eq:delta}
\end{align}

We search $\gamma$ for which the resulting balanced allocation $\bA$ is EF1.
Recall that a balanced allocation $\bA$ maximizes $\sum_{i\in N}\alpha_iv_i(A_i)$ if and only if $G^{(\balpha)}_{\bA}$ contains no negative-weight directed cycles.
In what follows, we write $G^{(\gamma)}_{\bA}$ instead of $G^{(\balpha)}_{\bA}$ for $\balpha=(1,\dots,1,\gamma,\dots,\gamma)$.
%
Define
\begin{align}
\mathcal{C} 
&= \{\gamma_1, \gamma_2, \ldots, \gamma_L\}\\
&= \left\{\tfrac{u_{1j}-u_{1j'}}{u_{2j}-u_{2j'}} \ \middle|\  j,j'\in M,~u_{1j}> u_{1j'}, u_{2j}> u_{2j'}\right\} \label{eq:critical-values}
\end{align}
to be the set of \emph{critical values} of $\gamma$, where $\delta < \gamma_1 < \gamma_2 < \cdots < \gamma_L < 1/\delta$. 
Intuitively, a critical value $\gamma$ represents a point at which a zero cycle appears in $G_{\bA}^{(\gamma)}$.
Note that $L=O(m^2)$ since there are at most $O(m^2)$ pairs of goods $(j,j')$.
We define the intervals between these critical values as $I_1 = [\delta, \gamma_1]$, $I_2 = [\gamma_1, \gamma_2]$, $\ldots$, $I_{L+1} = [\gamma_L, 1/\delta]$. We will denote $\gamma_0\coloneqq \delta$ and $\gamma_{L+1}\coloneqq 1/\delta$.
If $\mathcal{C}=\emptyset$, we define $L=0$.
\begin{lemma}\label{lem:critical-values}
For each $\ell\in[L+1]$, if $\bA$ maximizes $\sum_{i\in N_1}u_{1}(A_i)+\sum_{i\in N_2}\gamma u_{2}(A_i)$ for some $\gamma\in (\gamma_{\ell-1},\gamma_\ell)$, then it also maximizes the value for any $\gamma\in[\gamma_{\ell-1},\gamma_\ell]=I_\ell$.
\end{lemma}
\begin{appendixproof}[Proof of \Cref{lem:critical-values}]
Suppose that $\bA$ maximizes $\sum_{i\in N_1}u_{1}(A_i)+\sum_{i\in N_2}\gamma u_{2}(A_i)$ for some $\gamma^*\in (\gamma_{\ell-1},\gamma_\ell)$.
Assume to the contrary that there exists $\gamma'\in[\gamma_{\ell-1},\gamma_\ell]$ such that $\bA$ does not maximize $\sum_{i\in N_1}u_{1}(A_i)+\sum_{i\in N_2}\gamma' u_{2}(A_i)$. That is, there exists another balanced allocation $\bA'$ such that 
$\sum_{i\in N_1}u_{1}(A'_i)+\sum_{i\in N_2}\gamma' u_{2}(A'_i)>\sum_{i\in N_1}u_{1}(A_i)+\sum_{i\in N_2}\gamma' u_{2}(A_i)$.

Let $S=\bigcup_{i\in N_1}A_i$ and $S'=\bigcup_{i\in N_1}A'_i$, and define $S\setminus S'=\{j_1,j_2,\dots,j_s\}$ and $S'\setminus S=\{j'_1,j'_2,\dots,j'_s\}$.
By definition, the difference in objective value is
\begin{align}
\sum_{\kappa=1}^su_{1j'_\kappa}+\gamma' \sum_{\kappa=1}^su_{2j_\kappa}>\sum_{\kappa=1}^su_{1j_\kappa}+\gamma' \sum_{\kappa=1}^su_{2j'_\kappa}.
\end{align}
Thus, for some $\kappa^*\in[s]$, we have
\begin{align}
u_{1j'_{\kappa^*}}+\gamma' u_{2j_{\kappa^*}}>u_{1j_{\kappa^*}}+\gamma' u_{2j'_{\kappa^*}}.
\end{align}
On the other hand, since $\bA$ is optimal for $\gamma^*$, it must be that
\begin{align}
u_{1j'_{\kappa^*}}+\gamma^* u_{2j_{\kappa^*}}\le u_{1j_{\kappa^*}}+\gamma^* u_{2j'_{\kappa^*}},
\end{align}
because otherwise swapping $j_{\kappa^*}$ and $j'_{\kappa^*}$ in $\bA$ would strictly improve the objective at $\gamma^*$, contradicting optimality.

Therefore, as $\gamma$ varies from $\gamma^*$ to $\gamma'$, the comparison between $u_{1 j_{\kappa^*}'} + \gamma u_{2 j_{\kappa^*}}$ and $u_{1 j_{\kappa^*}} + \gamma u_{2 j_{\kappa^*}'}$ reverses. By continuity, there exists some $\gamma''$ between $\gamma^*$ and $\gamma'$ (with $\gamma''\ne \gamma'$) such that
\begin{align}
u_{1j'_{\kappa^*}}+\gamma'' u_{2j_{\kappa^*}}= u_{1j_{\kappa^*}}+\gamma'' u_{2j'_{\kappa^*}}.
\end{align}
Thus, we have $\gamma''=\frac{u_{1j_{\kappa^*}}-u_{1j'_{\kappa^*}}}{u_{2j_{\kappa^*}}-u_{2j'_{\kappa^*}}}$, meaning that $\gamma''$ is a critical value.
However, this is a contradiction since $\gamma_{\ell-1}<\gamma''<\gamma_{\ell}$.
Therefore, $\bA$ must maximize the objective for all $\gamma$ in $[\gamma_{\ell-1}, \gamma_\ell]$.
\end{appendixproof}
For each interval $I_\ell$, let $\bA^{(\ell)}$ be a balanced allocation that maximizes $\sum_{i\in N_1}u_{1}(A_i)+\sum_{i\in N_2}\gamma u_{2}(A_i)$ for all $\gamma \in I_\ell$.
Let $S^{(\ell)}\coloneqq\bigcup_{i\in N_1}A^{(\ell)}_i$ and $T^{(\ell)}\coloneqq\bigcup_{i\in N_2}A^{(\ell)}_i$.

For each interval $I_\ell$ and for each $\gamma\in I_\ell$, we determine the potentials $(\bq^{(\gamma)},\bp^{(\gamma)}) \in \mathbb{R}_+^N \times \mathbb{R}_+^M$ according to the procedure described in \Cref{lem:optimal-dual}.
Here, we do not need to specify $\ell$ since the potentials do not depend on the choice of the allocation.
Since agents of the same type have identical valuations and identical weights, there exists a path of length $0$ in $G^{(\gamma)}_{\bA^{(\ell)}}$ between any two agents of the same type.
Thus, we have $q^{(\gamma)}_i = q^{(\gamma)}_{i'}$ for any $i, i' \in N_1$ and $q^{(\gamma)}_i = q^{(\gamma)}_{i'}$ for any $i, i' \in N_2$. 
We redistribute the goods in $S^{(\ell)}$ among the type-1 agents, and those in $T^{(\ell)}$ among the type-2 agents, using the round-robin procedure with respect to the price vector $\bp^{(\gamma)}$.
Let $(X^{(\ell,\gamma)}_1, \dots, X^{(\ell,\gamma)}_{n_1})$ and $(Y^{(\ell,\gamma)}_1, \dots, Y^{(\ell,\gamma)}_{n_2})$ denote the resulting allocations for type-1 and type-2 agents, respectively.
We assume that the round-robin is applied in the order of indices.
Then, we have the following relations.
\begin{lemma}\label{lem:rr-price}
For any $\ell\in[L+1]$ and $\gamma\in I_\ell$ we have
\begin{align}
\MoveEqLeft[4]
p^{(\gamma)}(X^{(\ell,\gamma)}_1) \ge \dots \ge p^{(\gamma)}(X^{(\ell,\gamma)}_{n_1})\ge
\hat{p}^{(\gamma)}(X^{(\ell,\gamma)}_1) \ge \dots \ge \hat{p}^{(\gamma)}(X^{(\ell,\gamma)}_{n_1}),
\end{align}
and 
\begin{align}
\MoveEqLeft[4]
p^{(\gamma)}(Y^{(\ell,\gamma)}_1) \ge \dots \ge p^{(\gamma)}(Y^{(\ell,\gamma)}_{n_2})
\ge
\hat{p}^{(\gamma)}(Y^{(\ell,\gamma)}_1) \ge \dots \ge \hat{p}^{(\gamma)}(Y^{(\ell,\gamma)}_{n_2}).
\end{align}
\end{lemma}
\begin{appendixproof}[Proof of \Cref{lem:rr-price}]
For simplicity, we omit the superscripts $(\ell)$, ${(\gamma)}$ and $(\ell,\gamma)$.
By symmetry, it suffices to prove the inequalities for $(X_1,\dots,X_{n_1})$.

Let $S=\{s_1,s_2,\dots,s_{kn_1}\}$ where $p_{s_1}\ge p_{s_2}\ge\dots\ge p_{s_{kn_1}}$.
Then, we have $p(X_i)=\sum_{\iota=0}^{k-1}p_{s_{\iota n_1+i}}$ and $\hat{p}(X_i)=\sum_{\iota=1}^{k-1}p_{s_{\iota n_1+i}}$.
Therefore, for all $i,i'\in N_1$ with $i<i'$, we have $p(X_i)\ge p(X_{i'})$ and $\hat{p}(X_i)\ge \hat{p}(X_{i'})$.
Moreover, we have
\begin{align}
p(X_{n_1})
&=\sum_{\iota=0}^{k-1}p_{s_{\iota n_1+n_1}}
=\sum_{\iota=1}^{k}p_{s_{\iota n_1}}\\
&\ge \sum_{\iota=1}^{k-1}p_{s_{\iota n_1}}
\ge \sum_{\iota=1}^{k-1}p_{s_{\iota n_1+1}}
=\hat{p}(X_{1}).
\end{align}
This completes the proof.
\end{appendixproof}
Let $\hat{\bA}^{(\ell,\gamma)}=(X^{(\ell,\gamma)}_1, \dots, X^{(\ell,\gamma)}_{n_1}, Y^{(\ell,\gamma)}_1, \dots, Y^{(\ell,\gamma)}_{n_2})$.
We will omit the superscripts $(\gamma)$, $(\ell)$ and $(\ell,\gamma)$ when they are clear from the context.

For each interval $I_\ell$, we first check whether $\hat{\bA}^{(\ell,\gamma)}$ is EF1 for $\gamma = \gamma_{\ell-1}$ and $\gamma = \gamma_\ell$.
If either allocation is EF1, we immediately return it as an EF1 and fPO balanced allocation.
If not, we examine the following two conditions for each $\ell \in [L+1]$ and $\gamma \in \{\gamma_{\ell-1},\gamma_\ell\}$:
\begin{description}
    \item[(a)] $p^{(\gamma)}(X^{(\ell,\gamma)}_{n_1}) \ge \hat{p}^{(\gamma)}(Y_1^{(\ell,\gamma)})$,
    \item[(b)] $p^{(\gamma)}(Y^{(\ell,\gamma)}_{n_2}) \ge \hat{p}^{(\gamma)}(X^{(\ell,\gamma)}_1)$.
\end{description}
We will show in \Cref{lem:2types-conserved} that at least one of these two conditions holds for any $(\ell, \gamma)$.
Furthermore, if both conditions hold for some $(\ell, \gamma)$, then $\hat{\bA}^{(\ell,\gamma)}$ is EF1 (\Cref{lem:2types-pEF1}).
Since we have assumed that neither endpoint allocation is EF1, exactly one of the two conditions holds at each endpoint.
Additionally, we will show in \Cref{lem:2types-endpoints} that condition (a) holds at $(\ell, \gamma) = (1, \delta)$ and condition (b) holds at $(\ell, \gamma) = (L+1, 1/\delta)$.
Therefore, as we move through the intervals, at least one of the following two situations must occur:
\begin{itemize}
    \item Case 1: For some $\ell\in[L+1]$, condition (a) holds at $(\ell, \gamma_{\ell-1})$ and condition (b) holds at $(\ell, \gamma_\ell)$.
    \item Case 2: For some $\ell\in[L]$, condition (a) holds at $(\ell, \gamma_\ell)$ and condition (b) holds at $(\ell+1, \gamma_\ell)$.
\end{itemize}
We can find such an $\ell$ in polynomial time by checking the conditions for each $\ell\in[L+1]$ and $\gamma\in\{\gamma_{\ell-1},\gamma_\ell\}$.

If we encounter Case 1, there exists $\gamma^* \in I_\ell$ such that $\hat{\bA}^{(\ell,\gamma^*)}$ satisfies both conditions (a) and (b) (\Cref{lem:2types-case1}). Therefore, $\hat{\bA}^{(\ell,\gamma^*)}$ is EF1 and fPO by \Cref{lem:2types-pEF1}.
We search for such a $\gamma^*$ by tracking the changes as we continuously increase $\gamma$ within the interval $I_\ell$.
Since all agents of the same type share the same $q^{(\gamma)}_i$, there is a shortest path from the root $r$ to each node in $N \cup M$ that uses at most four edges in $G_{\bA^{(\ell)}}$. There are at most $O(m^2)$ such paths. The length of each path is a linear function of $\gamma$, so each $-q^{(\gamma)}_i$ and $p^{(\gamma)}_j$ can be represented as the minimum of $O(m^2)$ linear functions in $\gamma$. Thus, $-q^{(\gamma)}_i$ and $p^{(\gamma)}_j$ are piecewise linear functions with at most $O(m^2)$ segments.
This structure allows us to efficiently track the changes in the potentials as $\gamma$ varies, enabling us to find a suitable $\gamma^*$ that satisfies both conditions (a) and (b) in polynomial time.

If we encounter Case 2, we construct a sequence of allocations by successively exchanging one good from a type-1 agent with one good from a type-2 agent at each step, transitioning from $\hat{\bA}^{(\ell,\gamma_\ell)}$ to $\hat{\bA}^{(\ell+1,\gamma_\ell)}$. Specifically, at each step, we select a pair of goods $(j_1, j_2)$ with $j_1$ assigned to a type-1 agent and $j_2$ assigned to a type-2 agent, and swap their assignments.
Then, we redistribute the goods among type-1 agents and those among type-2 agents using the round-robin procedure with respect to the price vector $\bp^{(\gamma_\ell)}$.
By repeating such single-good exchanges, we can transform $\hat{\bA}^{(\ell,\gamma_\ell)}$ into $\hat{\bA}^{(\ell+1,\gamma_\ell)}$.
We will show that during this process, we reach an allocation that satisfies both conditions (a) and (b), thus yielding an EF1 and fPO balanced allocation (\Cref{lem:2types-case2-fPO,lem:2types-case2-p}).

\begin{fullpaper}
Our algorithm is summarized in \Cref{alg:2types}.
\begin{algorithm}[ht]
\caption{Two types case}\label{alg:2types}
\begin{algorithmic}[1]
\STATE Compute $\delta$ in \eqref{eq:delta} and critical values $\gamma_1,\ldots,\gamma_L$ in \eqref{eq:critical-values}
\FORALL{$\ell\in[L+1]$}\label{line:case0}
    \STATE Compute $\bA^{(\ell)}$ that maximizes $\sum_{i\in N_1}u_{1}(A_i)+\gamma\sum_{i\in N_2}u_{2}(A_i)$ for $\gamma=(\gamma_{\ell-1}+\gamma_{\ell})/2$
    \STATE \textbf{if} { $\hat{\bA}^{(\ell,\gamma_{\ell-1})}$ is EF1 }
        \textbf{then} \textbf{Return} $\hat{\bA}^{(\ell,\gamma_{\ell-1})}$
    \STATE \textbf{if} { $\hat{\bA}^{(\ell,\gamma_{\ell})}$ is EF1 }
        \textbf{then} \textbf{Return} $\hat{\bA}^{(\ell,\gamma_{\ell})}$
\ENDFOR
\IF { $\exists \ell\in [L+1]$ s.t. (a) holds at $(\ell, \gamma_{\ell-1})$ and (b) holds at $(\ell, \gamma_\ell)$}
    \STATE Compute $\gamma^* \in I_\ell$ such that $\hat{\bA}^{(\ell,\gamma^*)}$ satisfies both (a) and (b) by tracking the changes during the process
    \STATE \textbf{Return} $\hat{\bA}^{(\ell,\gamma^*)}$
\ELSE
    \STATE Let $\ell\in[L+1]$ be such that (a) holds at $(\ell, \gamma_\ell)$ and (b) holds at $(\ell+1, \gamma_\ell)$
    \STATE Let $S\gets S^{(\ell)}$ and $T\gets T^{(\ell)}$
    \WHILE { $S\ne S^{(\ell+1)}$ }
        \STATE Let $j_1\in S\cap T^{(\ell+1)}$ and $j_2\in T\cap S^{(\ell+1)}$
        \STATE $S\gets S\setminus\{j_1\}\cup\{j_2\}$;\quad $T\gets T\setminus\{j_2\}\cup\{j_1\}$
        \STATE Construct allocation $\hat{\bA}$ by redistributing goods in $S$ among type-1 agents and those in $T$ among type-2 agents using the round-robin procedure with respect to the price vector $\bp^{(\gamma_\ell)}$
        \STATE \textbf{if} {$\hat{\bA}$ is EF1}
            \textbf{then} \textbf{Return} $\hat{\bA}$
    \ENDWHILE
\ENDIF
\end{algorithmic}
\end{algorithm}
\end{fullpaper}

We now proceed to the proofs.
We first show that if $\hat{\bA}$ satisfies the conditions (a) and (b), then it is EF1.
\begin{lemma}\label{lem:2types-pEF1}
Let $\ell\in[L+1]$ and $\gamma\in I_\ell$.
If $p^{(\gamma)}(X^{(\ell,\gamma)}_{n_1})\ge \hat{p}^{(\gamma)}(Y^{(\ell,\gamma)}_1)$, then the envy from a type-1 agent to a type-2 agent can be eliminated by removing one good.
Similarly, if $p^{(\gamma)}(Y^{(\ell,\gamma)}_{n_2})\ge \hat{p}^{(\gamma)}(X^{(\ell,\gamma)}_1)$, then the envy from a type-2 agent to a type-1 agent can be eliminated by removing one good.
In addition, if both conditions hold, then the allocation is EF1.
\end{lemma}
\begin{appendixproof}[Proof of \Cref{lem:2types-pEF1}]
For simplicity, we omit the superscripts ${(\gamma)}$ and $(\ell,\gamma)$.
By symmetry, we only show the first part.
Suppose that $p(X_{n_1})\ge \hat{p}(Y_1)$.
Then, by \Cref{lem:rr-price}, we have
\begin{align}
&p(X_1)\ge \dots\ge p(X_{n_1})\ge \hat{p}(Y_1)\ge \dots\ge \hat{p}(Y_{n_2}).
\end{align}
Thus, we have $p(X_i)\ge \hat{p}(Y_{i'})$ for any $i\in N_1$ and $i'\in N_2$.
Moreover, by \Cref{lem:rr-price}, we have $p(X_i)\ge \hat{p}(X_{i'})$ for any $i,i'\in N_1$.
Therefore, by \Cref{lem:pEF1}, we can eliminate the envy from a type-1 agent to a type-2 agent by removing one good.
\end{appendixproof}

Next, we show that at least one of the conditions (a) or (b) holds for any $\ell\in[L+1]$ and $\gamma\in I_\ell$.
\begin{lemma}\label{lem:2types-conserved}
For any $\ell\in[L+1]$ and $\gamma\in I_\ell$, at least one of $p^{(\gamma)}(X^{(\ell,\gamma)}_{n_1})\ge \hat{p}^{(\gamma)}(Y_1^{(\ell,\gamma)})$ or $p^{(\gamma)}(Y^{(\ell,\gamma)}_{n_2})\ge \hat{p}^{(\gamma)}(X^{(\ell,\gamma)}_1)$ holds.
\end{lemma}
\begin{appendixproof}[Proof of \Cref{lem:2types-conserved}]
For simplicity, we omit the superscript $(\ell,\gamma)$.
Suppose that we have $p(X_{n_1})<\hat{p}(Y_1)$. Then, we have
\begin{align}
\hat{p}(X_1)
\le p(X_{n_1})
<\hat{p}(Y_1)
\le p(Y_{n_2}),
\end{align}
where the first and the last inequalities hold by \Cref{lem:rr-price}.
Hence, we obtain $p(Y_{n_2})\ge \hat{p}(X_1)$.
\end{appendixproof}

We show that if $\hat{\bA}^{(1,\delta)}$ is not EF1, then condition (a) holds for $(\ell,\gamma)=(1,\delta)$, and if $\hat{\bA}^{(L+1,1/\delta)}$ is not EF1, then condition (b) holds for $(\ell,\gamma)=(L+1,1/\delta)$.
\begin{lemma}\label{lem:2types-endpoints}
If $\hat{\bA}^{(1,\delta)}$ is not EF1, then $p^{(\delta)}(X^{(1,\delta)}_{n_1})\ge \hat{p}^{(\delta)}(Y_1^{(1,\delta)})$.
Also, if $\hat{\bA}^{(L+1,1/\delta)}$ is not EF1, then $p^{(1/\delta)}(Y^{(L+1,1/\delta)}_{n_2})\ge \hat{p}^{(1/\delta)}(X^{(L+1,1/\delta)}_1)$.
\end{lemma}
\begin{appendixproof}[Proof of \Cref{lem:2types-endpoints}]
By symmetry, it suffices to show the first part.
Recall that $\bA^{(1)}$ is an fPO balanced allocation that maximizes $\sum_{i\in N_1}u_{1}(A_i) + \delta\sum_{i\in N_2}u_{2}(A_i)$.
Then, for any good $j\in S^{(1)}$ and $j'\in T^{(1)}$, we have
\begin{align}
0
&\le (u_{1j}+\delta u_{2j'})-(u_{1j'}+\delta u_{2j})\\
&\le u_{1j}-u_{1j'}+\delta u_{2j'}.
\end{align}
Since $\delta u_{2j'}<|u_{1j}-u_{1j'}|$ if $u_{1j}\ne u_{1j'}$ by the definition of $\delta$,
this implies that $u_{1j}\ge u_{1j'}$ for any good $j\in S^{(1)}$ and $j'\in T^{(1)}$.
Thus, $u_1(X^{(1,\delta)}_i)\ge u_1(Y^{(1,\delta)}_{i'})$ for all $i\in [n_1]$ and $i'\in [n_2]$ by $|X_i|=|Y_{i'}|=k$, i.e., there is no envy from a type-1 agent to a type-2 agent.

Suppose that $p^{(1,\delta)}(Y^{(1,\delta)}_{n_2})\ge \hat{p}^{(1,\delta)}(X^{(1,\delta)}_1)$. Then, by \Cref{lem:2types-pEF1}, the envy from type-2 agent to a type-1 agent can be eliminated by removing one good. Since there is no envy from a type-1 agent to a type-2 agent, the allocation $\hat{\bA}^{(1,\delta)}$ must be EF1. 

If $\hat{\bA}^{(1,\delta)}$ is not EF1, we have $p^{(1,\delta)}(Y^{(1,\delta)}_{n_2})< \hat{p}^{(1,\delta)}(X^{(1,\delta)}_1)$. 
Hence, $p^{(1,\delta)}(Y^{(1,\delta)}_{n_2})< \hat{p}^{(1,\delta)}(X^{(1,\delta)}_1)$.
By \Cref{lem:2types-conserved}, this implies that $p^{(1,\delta)}(X^{(1,\delta)}_{n_1})\ge \hat{p}^{(1,\delta)}(Y_1^{(1,\delta)})$.
\end{appendixproof}

We establish the existence of the desired $\gamma^*$ in Case 1.
As a preliminary step, we show the continuity of the prices received by each agent.
\begin{lemma}\label{lem:2types-continuous}
For each $\ell\in[L+1]$, $i\in N_1$, and $i'\in N_2$, the values
$p^{(\gamma)}(X^{(\ell,\gamma)}_{i})$,
$\hat{p}^{(\gamma)}(X^{(\ell,\gamma)}_{i})$,
$p^{(\gamma)}(Y_{i'}^{(\ell,\gamma)})$, and 
$\hat{p}^{(\gamma)}(Y_{i'}^{(\ell,\gamma)})$ 
are continuous with respect to $\gamma$ over $I_\ell$.
\end{lemma}
\begin{appendixproof}[Proof of \Cref{lem:2types-continuous}]
Recall that in the interval $I_\ell$, the goods $S^{(\ell)}$ and $T^{(\ell)}$ are allocated to types-1 agents and type-2 agents, respectively, using the round-robin procedure with respect to the price vector $p^{(\gamma)}$.

For each $j\in M$, the price $p_j^{(\ell,\gamma)}$ is a piecewise linear and continuous function of $\gamma$ on $I_\ell$, as it is defined as the minimum of finitely many linear functions in $\gamma$ by \Cref{lem:optimal-dual}.
The round-robin allocation of goods among agents of the same type depends only on the ordering of good prices. As $\gamma$ varies, the allocation changes only at points where two good prices coincide. Between such points, the sets $X_i^{(\ell,\gamma)}$ and $Y_{i'}^{(\ell,\gamma)}$ remain constant, and thus their total prices are continuous in $\gamma$. At the breakpoints, the sets may change, but since the exchanged goods have equal prices, the total price assigned to each agent remains continuous.
Therefore, the values $p^{(\gamma)}(X^{(\ell,\gamma)}_{i})$, $\hat{p}^{(\gamma)}(X^{(\ell,\gamma)}_{i})$, $p^{(\gamma)}(Y_{i'}^{(\ell,\gamma)})$, and $\hat{p}^{(\gamma)}(Y_{i'}^{(\ell,\gamma)})$ are continuous functions of $\gamma$ over the interval $I_\ell$.
\end{appendixproof}

Now we can prove the existence of $\gamma^*$ in Case 1.
\begin{lemma}\label{lem:2types-case1}
Take any $\ell$ in $[L+1]$.
Consider the case where
$p^{(\gamma_{\ell-1})}(X^{(\ell,\gamma_{\ell-1})}_{n_1})\ge \hat{p}^{(\gamma_{\ell-1})}(Y_1^{(\ell,\gamma_{\ell-1})})$ and 
$p^{(\gamma_\ell)}(Y^{(\ell,\gamma_\ell)}_{n_2})\ge \hat{p}^{(\gamma_\ell)}(X^{(\ell,\gamma_\ell)}_1)$.
Then, there exists $\gamma^*\in I_\ell$ such that 
$p^{(\gamma^*)}(X^{(\ell,\gamma^*)}_{n_1})\ge \hat{p}^{(\gamma^*)}(Y_1^{(\ell,\gamma^*)})$ and 
$p^{(\gamma^*)}(Y^{(\ell,\gamma^*)}_{n_2})\ge \hat{p}^{(\gamma^*)}(X^{(\ell,\gamma^*)}_1)$.
Moreover, such a $\gamma^*$ can be computed in polynomial time.
\end{lemma}
\begin{appendixproof}[Proof of \Cref{lem:2types-case1}]
Define the sets
\begin{align}
\Gamma_1 &= \left\{ \gamma \in I_\ell \ \middle|\ p^{(\gamma)}(X^{(\ell,\gamma)}_{n_1}) \ge \hat{p}^{(\gamma)}(Y_1^{(\ell,\gamma)}) \right\},
\ \text{and}\\
\Gamma_2 &= \left\{ \gamma \in I_\ell \ \middle|\ p^{(\gamma)}(Y^{(\ell,\gamma)}_{n_2}) \ge \hat{p}^{(\gamma)}(X^{(\ell,\gamma)}_1) \right\}.
\end{align}
By assumption, $\gamma_{\ell-1} \in \Gamma_1$ and $\gamma_\ell \in \Gamma_2$.
Since the values $p^{(\gamma)}(X^{(\ell,\gamma)}_{n_1})$, $\hat{p}^{(\gamma)}(Y_1^{(\ell,\gamma)})$, $p^{(\gamma)}(Y^{(\ell,\gamma)}_{n_2})$, and $\hat{p}^{(\gamma)}(X^{(\gamma)}_1)$ are continuous in $\gamma$, both $\Gamma_1$ and $\Gamma_2$ are closed subsets of $I_\ell$.
Let $\gamma^* = \max \Gamma_1$. 
If $\gamma^*=\gamma_\ell$, then we have $\gamma^*\in\Gamma_1\cap\Gamma_2$ by $\gamma_\ell\in\Gamma_2$. 
If $\gamma^*<\gamma_\ell$, then we have $(\gamma^*, \gamma_\ell] \subseteq \Gamma_2$ by \Cref{lem:2types-conserved}.
Since $\Gamma_2$ is closed, we have $[\gamma^*,\gamma_\ell]\subseteq \Gamma_2$.
Thus, $\gamma^* \in \Gamma_1 \cap \Gamma_2$ and satisfies the desired conditions.

Finally, we show that such $\gamma^*$ can be efficiently found by tracking the changes in the prices as we increase $\gamma$ within the interval $I_\ell$.
By \Cref{lem:optimal-dual}, the dual variables $-\bq^{(\gamma)}$ and $\bp^{(\gamma)}$ correspond to shortest path lengths in the graph $G^{(\gamma)}_{\bA^{(\ell)}}$. 
Since all agents of the same type share the same $q_i$, there is a shortest path from the root $r$ to each node in $N \cup M$ that uses at most four edges. There are at most $O(m^2)$ such paths. The length of each path is a linear function of $\gamma$, so each $-q^{(\gamma)}_i$ and $p^{(\gamma)}_j$ can be represented as the minimum of $O(m^2)$ linear functions in $\gamma$. Thus, $-q^{(\gamma)}_i$ and $p^{(\gamma)}_j$ are piecewise linear functions with at most $O(m^2)$ segments.
Furthermore, the round-robin redistribution of goods among agents depends only on the price vector $\bp^{(\gamma)}$. The allocation changes only when the prices of two goods cross, which can happen at most once per segment for each pair of goods. Therefore, within the interval $I_\ell$, the total number of allocation changes is bounded by $O(m^2 n^2)$.
By enumerating all breakpoints where either the segments of potentials or the allocation changes, we can efficiently track the values of $p^{(\gamma)}(X^{(\ell,\gamma)}_{n_1})$ and $\hat{p}^{(\gamma)}(Y^{(\ell,\gamma)}_1)$ as we increase $\gamma$ within the interval $I_\ell$.
Therefore, we can compute the desired $\gamma^*$ in polynomial time.
\end{appendixproof}

We next show that, if we encounter Case 2, then we can find an EF1 and fPO balanced allocation by successively exchanging one good from a type-1 agent with one good from a type-2 agent.
Let $\ell$ be the index chosen in the algorithm. 
Since both $\bA^{(\ell)}$ and $\bA^{(\ell+1)}$ maximize $\sum_{i\in N_1}u_{1}(A_i)+\gamma_\ell\sum_{i\in N_2}u_{2}(A_i)$, we have $q^{(\gamma_\ell)}_i+p^{(\gamma_\ell)}_j=u_{1j}$ for all $i\in N_1$ and $j\in S^{(\ell)}\cup S^{(\ell+1)}$, and $q^{(\gamma_\ell)}_i+p^{(\gamma_\ell)}_j=\gamma_\ell u_{2j}$ for all $i\in N_2$ and $j\in T^{(\ell)}\cup T^{(\ell+1)}$.
We observe that $\hat{\bA}$ obtained during the process of Case 2 is always fPO.
\begin{lemma}\label{lem:2types-case2-fPO}
Fix $\ell\in[L+1]$ and let $(\bq,\bp)=(\bq^{(\gamma_\ell)},\bp^{(\gamma_\ell)})$.
Suppose $M$ is partitioned into $S$ and $T$ with $S\subseteq S^{(\ell)}\cup S^{(\ell+1)}$ and $T\subseteq T^{(\ell)}\cup T^{(\ell+1)}$.
Let $\hat{\bA}=(X_1,\dots,X_{n_1},Y_1,\dots,Y_{n_2})$ be the balanced allocation obtained by the round-robin procedure with respect to $\bp$, assigning goods in $S$ to type-1 agents and goods in $T$ to type-2 agents.
Then, $\hat{\bA}$ is fPO.
\end{lemma}
\begin{appendixproof}[Proof of \Cref{lem:2types-case2-fPO}]
For each $i\in N_1$, we have $q_i+p_j=u_{1j}$ for all $j\in X_i\subseteq S^{(\ell)}\cup S^{(\ell+1)}$.
Similarly, for each $i\in N_2$, we have $q_{i}+p_j=\gamma_\ell u_{2j}$ for all $j\in Y_{i}\subseteq T^{(\ell)}\cup T^{(\ell+1)}$.
Thus, by the complementary slackness condition, $\hat{\bA}$ maximizes $\sum_{i\in N_1}u_{1}(A_i)+\gamma_\ell\sum_{i\in N_2}u_{2}(A_i)$. Hence, $\hat{\bA}$ is fPO.
\end{appendixproof}

Next, we show that if $\hat{\bA}$ does not satisfy condition (b) at some step, then $\hat{\bA}$ at the next step satisfies condition (a).
\begin{lemma}\label{lem:2types-case2-p}
Let $(S,T)$ be a partition of $M$ and $\bp\in\mathbb{R}_+^N$.
Let $\hat{\bA}=(X_1,\dots,X_{n_1},Y_1,\dots,Y_{n_2})$ be the balanced allocation obtained by the round-robin procedure with respect to $\bp$, assigning goods in $S$ to type-1 agents and goods in $T$ to type-2 agents.
Suppose $j_1\in S$ and $j_2\in T$, and consider the new partition $(S',T')=(S\setminus\{j_1\}\cup\{j_2\},\, T\setminus\{j_2\}\cup\{j_1\})$, and let $\hat{\bA}'=(X'_1,\dots,X'_{n_1},Y'_1,\dots,Y'_{n_2})$ be the balanced allocation obtained by round-robin with respect to $(S',T')$ and $\bp$.
Then, if $p(Y_{n_2})<\hat{p}(X_1)$, we have $p(X'_{n_1})\ge \hat{p}(Y'_1)$.
\end{lemma}
\begin{appendixproof}[Proof of \Cref{lem:2types-case2-p}]
We first prove that 
\begin{align}
p(X'_{n_1})\ge \hat{p}(X_1).\label{eq:case2-p1}
\end{align}
Let $S=\{s_1,s_2,\dots,s_{kn_1}\}$ and $S'=\{s'_1,s'_2,\dots,s'_{kn_1}\}$
where $p_{s_1}\ge p_{s_2}\ge \dots\ge p_{s_{kn_1}}$ and $p_{s'_1}\ge p_{s'_2}\ge \dots\ge p_{s'_{kn_1}}$.
Then, we have $\hat{p}(X_1)=\sum_{\iota=1}^{k-1} p_{s_{\iota n_1+1}}$ and $p(X'_{n_1})=\sum_{\iota=1}^{k} p_{s'_{\iota n_1}}$.
Since $S'=S\setminus\{j_1\}\cup\{j_2\}$, we have $p_{s'_\eta}\ge p_{s_{\eta+1}}$ for all $\eta\in [kn_1-1]$. Hence, we obtain
\begin{align}
p(X'_{n_1})
=\sum_{\iota=1}^{k} p_{s'_{\iota n_1}}
\ge \sum_{\iota=1}^{k-1} p_{s'_{\iota n_1}}
\ge \sum_{\iota=1}^{k-1} p_{s_{\iota n_1+1}}
=\hat{p}(X_1).
\end{align}

Next, we prove that
\begin{align}
p(Y_{n_2})\ge \hat{p}(Y'_1).\label{eq:case2-p2}
\end{align}
Let $T=\{t_1,t_2,\dots,t_{kn_2}\}$ and $T'=\{t'_1,t'_2,\dots,t'_{kn_2}\}$
where $p_{t_1}\ge p_{t_2}\ge \dots\ge p_{t_{kn_2}}$ and $p_{t'_1}\ge p_{t'_2}\ge \dots\ge p_{t'_{kn_2}}$.
Then, we have $\hat{p}(Y_{n_2})=\sum_{\iota=1}^{k} p_{t_{\iota n_2}}$ and $\hat{p}(Y'_{1})=\sum_{\iota=1}^{k-1} p_{s'_{\iota n_2+1}}$.
Since $T'=T\setminus\{j_2\}\cup\{j_1\}$, we have $p_{t_\eta}\ge p_{t'_{\eta+1}}$ for all $\eta\in [kn_2-1]$. Hence, we obtain
\begin{align}
p(Y_{n_2})
=\sum_{\iota=1}^{k} p_{t_{\iota n_2}}
\ge \sum_{\iota=1}^{k-1} p_{t_{\iota n_2}}
\ge \sum_{\iota=1}^{k-1} p_{t'_{\iota n_2+1}}
=\hat{p}(Y'_1).
\end{align}

By combining \eqref{eq:case2-p1} and \eqref{eq:case2-p2}, we have
\begin{align}
\hat{p}(Y'_1)\le p(Y_{n_2})<\hat{p}(X_1)\le p(X'_{n_1}),
\end{align}
if $p(Y_{n_2})<\hat{p}(X_1)$.
Therefore, $p(Y_{n_2})<\hat{p}(X_1)$ implies $p(X'_{n_1})\ge \hat{p}(Y'_1)$.
\end{appendixproof}
This lemma implies that if we encounter Case 2, we can find an EF1 and fPO balanced allocation by successively exchanging goods between type-1 and type-2 agents until conditions (a) and (b) are satisfied.
This is because $\hat{\bA}=\hat{\bA}^{(\ell,\gamma_\ell)}$ at the first step satisfies condition (a), and at the end we will have $\hat{\bA}=\hat{\bA}^{(\ell+1,\gamma_\ell)}$ and it satisfies condition (b). If an allocation satisfies both conditions (a) and (b), then it is EF1 by \Cref{lem:2types-pEF1}.

Now, we are ready to prove \Cref{thm:2types}.
\begin{proof}[Proof of \Cref{thm:2types}]
If $\hat{\bA}^{(\ell,\gamma)}$ is EF1 for some $\ell\in[L+1]$ and $\gamma\in\{\gamma_{\ell-1},\gamma_\ell\}$, then the algorithm returns it as an EF1 and fPO balanced allocation. 
Otherwise, the algorithm finds an EF1 and fPO balanced allocation via Case 1 or Case 2, depending on which situation occurs.
For Case 1, it finds $\gamma^*\in I_\ell$ such that $\hat{\bA}^{(\ell,\gamma^*)}$ satisfies both conditions (a) and (b), and returns it as an EF1 and fPO balanced allocation by \Cref{lem:2types-case1}.
For Case 2, it finds an EF1 and fPO balanced allocation by \Cref{lem:2types-case2-fPO} and \Cref{lem:2types-case2-p}.
Thus, the algorithm always finds an EF1 and fPO balanced allocation.

For each $\ell\in[L+1]$, the allocation $\bA^{(\ell)}$ can be computed in polynomial time by the Hungarian algorithm. 
For each $\ell\in[L+1]$ and $\gamma\in I_\ell$, we can compute the potential $(\bq^{(\gamma)},\bp^{(\gamma)})$ in polynomial time by the Bellman--Ford algorithm by \Cref{lem:optimal-dual}.
We can then compute the allocation $\hat{\bA}^{(\ell,\gamma)}$ in polynomial time since the round-robin redistribution of goods takes $O(m)$ time. 
Since $L=O(m^2)$ and the critical values can be computed in $O(m^2)$ time, we can check whether $\hat{\bA}^{(\ell,\gamma)}$ is EF1 for each $\ell\in[L+1]$ and $\gamma\in\{\gamma_{\ell-1},\gamma_\ell\}$ in polynomial time.
For Case 1, we can find a desired $\gamma$ in polynomial time by \Cref{lem:2types-case1}. Hence, the total time complexity for Case 1 is polynomial.
For Case 2, we can find an EF1 and fPO balanced allocation in polynomial time since we only need to perform a sequence of single-good exchanges, and the number of such exchanges is at most $O(m)$.
\end{proof}

\section{Conclusion}
In this paper, we addressed the problem of finding a fair and efficient allocation of indivisible goods under the balanced constraint, where each agent receives the same number of goods. 
Our main contribution is to affirmatively resolve the existence and polynomial-time computability of allocations that are simultaneously EF1 and fPO in two important cases: when agents have personalized bivalued valuations and when there are at most two types of agents. We developed novel polynomial-time algorithms for both scenarios, leveraging techniques from maximum-weight perfect matching.

A natural next step is to extend our results beyond two types of agents to the general case of three or more types, but this becomes much more complex.


\section{Acknowledgments}
This work was partially supported by the joint project of Kyoto University and Toyota Motor Corporation, titled ``Advanced Mathematical Science for Mobility Society'' and supported by JST ERATO Grant Number JPMJER2301, JST PRESTO Grant Number JPMJPR2122, 
JSPS KAKENHI Grant Numbers JP23K19956 and JP25K00137, Japan.

\bibliographystyle{plain}
\bibliography{fair}


\begin{fullpaper}
\clearpage
\appendix
\section{Omitted Proof}
    
\begin{proof}[Proof of \Cref{lem:optimal-dual}]
We remark that the weighted bipartite graph $G^{(\balpha)}_{\bA}$ does not contain a negative directed cycle, since otherwise we can construct a balanced allocation $\bA'$ with $\sum_{i\in N}\alpha_iv_i(A'_i)>\sum_{i\in N}\alpha_iv_i(A_i)$ by transferring goods along the cycle.

Let $q_i$ be the shortest length of $r$--$i$ path for each $i\in N$ and $-p_j$ as the shortest length of $r$--$j$ path for each $j\in M$ in $G^{(\balpha)}_{\bA}$.
Then, for each $(i,j)\in N\times M$, we have $-p_j\le q_i-\alpha_iv_{ij}$, since there is a path from $r$ to $j$ that follows a shortest path from $r$ to $i$, and then use the edge $(i,j)$ with weight $-\alpha_iv_{ij}$.
Thus, the solution $(\bq,\bp)$ is feasible to \eqref{eq:dual}.
In addition, for each $(i,j)\in N\times M$ with $j\in A_i$, we have $q_i\le -p_j+\alpha_iv_{ij}$, since there is a path from $r$ to $j$ that follows a shortest path from $r$ to $i$, and then use the edge $(j,i)$ with weight $\alpha_iv_{ij}$.
This implies that $q_i+p_j=\alpha_iv_{ij}$ for all $(i,j)\in N\times M$ with $j\in A_i$.
Thus, $(\bq,\bp)$ is optimal solution of \eqref{eq:dual} by the complementary slackness.

Next, we observe that the solution $(\bq,\bp)$ is nonnegative.
For each $i\in N$, the value $q_i$ is nonnegative; otherwise, we can construct a negative cycle as follows: along a shortest path from $r$ to $i$, consider the portion of the path after the first vertex $j\in M$ following $r$, and combine this with the edge $(i,j)$. This contradicts the fact that there is no negative directed cycle in $G^{(\balpha)}_{\bA}$.
For each $j\in M$, the value $-p_j$ is nonpositive (and hence $p_j$ is nonnegative) since the path $(r,j)$ is of length $0$ and any shortest $r$--$j$ path is not longer than this.

Finally, suppose that $\bA'$ is a balanced allocation that also maximizes the weighted sum, i.e., $\sum_{i\in N}\alpha_iv_i(A'_i)=\sum_{i\in N}\alpha_iv_i(A_i)$.
Then, for each $(i,j)\in N\times M$ with $j\in A_i'$, we have $p_i+q_j=\alpha_iv_{ij}$, since otherwise
\begin{align}
\sum_{i\in N}\alpha_iv_i(A_i)
&=\sum_{j\in M}(q_{\bA(j)}+p_j)
=k\cdot\sum_{i\in N}q_i+\sum_{j\in M}p_j\\
&>\sum_{j\in M}(q_{\bA'(j)}+p_j)
=\sum_{i\in N}\alpha_iv_i(A'_i),
\end{align}
which is a contradiction.
Here, we write $\bA(j)$ and $\bA'(j)$ to denote the agent who receives $j$ in $\bA$ and $\bA'$, respectively.
Thus, for each $(i,j)\in N\times M$ with $j\in A_i$, there is a zero-cycle in $G^{(\balpha)}_{\bA}$ that contains $(i,j)$ and alternately uses an edge of the form $(j',\bA'(j'))$ and one of the form $(\bA(j'),j')$.
This implies that there is a path in $G^{(\balpha)}_{\bA'}$ from $i$ to $j$ of length at most $-\alpha_iv_{ij}$.
Hence, the shortest distance from $r$ to each vertex in $G^{(\balpha)}_{\bA'}$ is never longer than the corresponding shortest distance in $G^{(\balpha)}_{\bA}$.
By symmetry, an analogous statement holds for $G^{(\balpha)}_{\bA}$ relative to $G^{(\balpha)}_{\bA'}$ as well.
Therefore, we can conclude that $(\bq,\bp)=(\bq',\bp')$.
\end{proof}

\end{fullpaper}


\end{document}